\documentclass{article}
\usepackage{amsthm,amssymb,amsfonts,bbm}
\usepackage{amsmath}
\usepackage{color}
\definecolor{myblue}{rgb}{0,0,0.5}
\usepackage[bookmarksnumbered,colorlinks,linkcolor=myblue,citecolor=myblue,urlcolor=myblue,
  pdftitle={Classical and quantum partition bound and detector inefficiency},
  pdfauthor={S. Laplante, V. Lerays and J. Roland}
]{hyperref}

\usepackage{llncsdoc}

\usepackage{fullpage}

\usepackage[utf8x]{inputenc}

\newenvironment{myitemize}
{\vspace{-1mm}\begin{list}{$\bullet$}%
    {\setlength{\itemsep}{-3pt}}%
    {\setlength{\topsep}{0pt}}%
    {\setlength{\partopsep}{0pt}}%
    {\setlength{\parsep}{0pt}}%
}
{\end{list}%
\vspace{-3mm}
}
\newenvironment{myenumerate}
{ \vspace{-1mm}\begin{list}{\arabic{enumi}.}{\usecounter{enumi}}
    \setlength{\itemsep}{-3pt}%
    \setlength{\topsep}{0pt}%
    \setlength{\partopsep}{0pt}%
    \setlength{\parsep}{0pt}%
}
{\end{list}%
\vspace{-3mm}}

\newcommand{\ket} [1] {\vert #1 \rangle}

\newtheorem{thm}{Theorem}
\newtheorem{lemma}[thm]{Lemma}
\newtheorem{defn}{Definition}
\newtheorem{cor}[thm]{Corollary}
\newtheorem{claim}{Claim}
\newtheorem{proposition}[thm]{Proposition}

\newcommand{\abs}[1]{\mid\! #1\! \mid}

\renewcommand{\L}{{\cal L}}

\newcommand{\Q}{{\cal Q}}

\newcommand{\Real}{\mathbb R}

\newcommand{\A}{\mathcal{A}}
\newcommand{\B}{\mathcal{B}}

\newcommand{\X}{\mathcal{X}}
\newcommand{\Y}{\mathcal{Y}}
\newcommand{\Z}{\mathcal{Z}}

\newcommand{\Prob}{{\rm Prob}}

\newcommand{\opt}{\mathrm{opt}}
\newcommand{\botA}{{\bot A}}

\newcommand{\ind}{\mathrm{ind}}
\newcommand{\eff}{\mathrm{\bf eff}}

\newcommand{\nc}{\mathrm{nc}}
\newcommand{\effl}{\mathrm{eff}^\nc}
\newcommand{\prt}{\mathrm{prt}}

\newcommand{\HM}{\mathrm{HM}}
\newcommand{\KV}{\mathrm{KV}}

\newcommand{\bp}{\mathbf{p}}
\newcommand{\bl}{\mathbf{\ell}}
\newcommand{\bq}{\mathbf{q}}

\newcommand{\enlever}[1]{}

\title{Classical and quantum partition bound and detector inefficiency \footnote{Extended version of the previously published article: Sophie Laplante, Virginie Lerays and J\'er\'emie Roland, ``Classical and quantum partition bound and detector efficiency", Proceedings of the 39th International Colloquium on Automata, Languages and Programming, 2012}%
}
\author{Sophie Laplante\\ LIAFA, Universit\'e Paris Diderot Paris 7 \and Virginie Lerays\\ LIAFA, Universit\'e Paris Diderot Paris 7   \and J{\'e}r{\'e}mie Roland\\ ULB, QuIC, Ecole Polytechnique de Bruxelles }

\date{}
\begin{document}

\maketitle

\thispagestyle{empty}

\begin{abstract}
We study randomized and quantum efficiency lower bounds in communication complexity.
These arise from the study of zero-communication protocols in which players
are allowed to abort. 
Our scenario is inspired by the physics 
setup of Bell experiments, where two players share a predefined 
entangled state but are not allowed to communicate.  
Each is given a measurement as input, which they
perform on their share of the system. The outcomes of the measurements
should follow a distribution predicted by quantum mechanics;
however, in practice, the detectors may fail to produce an output
in some of the runs.
The efficiency of the experiment is the probability
that the experiment succeeds (neither of the detectors fails).  

When the players share a quantum state, this gives rise to
a new bound on quantum communication complexity ($\eff^*$)
that subsumes the factorization norm.
When players share randomness instead of a quantum state,
the efficiency bound ($\eff$), 
coincides with the partition bound of Jain and Klauck.
This is one of the strongest lower bounds known for randomized
communication complexity, which subsumes all the known combinatorial and 
algebraic methods including the rectangle (corruption) bound,
the factorization norm, and discrepancy.  

The lower bound is formulated as a convex optimization
problem.  In practice, the dual form is more feasible 
to use,  and we show that 
it amounts to constructing an explicit Bell inequality (for $\eff$) or Tsirelson
inequality
(for $\eff^*$).  We give an example of a quantum distribution where
the violation can be exponentially bigger than the previously
studied class of normalized Bell inequalities.

For one-way communication, we show that the quantum one-way partition
bound is tight for classical communication with shared entanglement
up to arbitrarily small error.

\enlever{
Finally, an important goal in physics is to devise robust Bell 
experiments that are impervious to noise and detector 
inefficiency.  We give a general tradeoff between
communication, Bell inequality violation, and detector efficiency.}
\end{abstract}

\setcounter{page}{1}

\section{Introduction}

How are Bell tests related to communication complexity? At a high level, both involve two distant players, Alice and Bob, who receive inputs $x,y$, respectively, and produce outputs $a,b$ according to some distribution $p(a,b|x,y)$. The goal of a Bell test is to show that a given distribution $p(a,b|x,y)$ (typically arising from performing measurements on a shared entangled state) cannot result from a so-called local hidden variable model, which we will call here \emph{local protocol} (or zero-communication protocol) for simplicity. A local protocol is a protocol where Alice and Bob may use shared randomness only but no communication (nor shared entanglement). In practice they may also abort. The interesting quantity for us is the {\em efficiency}, that is, the probability that the players do not abort. A lower efficiency makes it easier to reproduce the distribution using a local protocol, and a meaningful measure is therefore the maximum efficiency such that a local protocol for $p(a,b|x,y)$ exists. In the communication complexity model, Alice and Bob have inputs $x$ and~$y$ respectively and must minimize the communication between them in order to solve a distributed task (or equivalently output $a,b$ according to distribution $p(a,b|x,y)$). Both are measures of how far a given distribution is from the set of
local (zero-communication) distributions.

Massar and Buhrman \textit{et al.}~\cite{M01, BHMR03} 
described how a communication protocol gives rise to a local protocol
where the players can abort: 
if there is a $c$-bit communication protocol
where Alice and Bob output $a,b$ with distribution $p(a,b|x,y)$ when Alice's
input is $x$ and Bob's input is $y$,  then there is a local protocol
that outputs according to $\bp$ (conditioned
on the run not being aborted) whose probability of not aborting is $2^{-c}$. Both players use the shared randomness to guess a transcript,
and if they disagree with the transcript, they abort.  Otherwise
they output according to the transcript.  

In this paper we prove a much stronger relation between communication protocols and the notion of efficiency in Bell tests, and provide applications both in communication complexity and in Bell inequality violations for quantum distributions. 
More precisely, we show that {\bf the efficiency bound} (i.e., the maximum efficiency of any local protocol that simulates $\bp$)  is in fact {\bf equal to the partition bound} in classical communication complexity~\cite{JK10}. The partition bound is important because it is one of the `strongest' known classical communication lower bound techniques. Moreover
we obtain a {\bf strong new bound for quantum communication complexity}, which is at least as strong as all previously known 
lower bound techniques for quantum communication complexity.  
We show that the one-way version of the quantum efficiency bound is tight.

We show that the efficiency bound is equivalent to finding 
{\bf Bell inequalities that are resistant to the detection loophole}, exhibiting an unexpected connection between these notions. This enables us to exhibit  
a quantum distribution arising from measurements on an $n$-dimensional shared 
quantum state, but which provides exponential Bell violations.

\subsection{Communication complexity and the partition bound}
Recently, Jain and Klauck~\cite{JK10} proposed a new lower bound on 
randomized communication complexity which subsumes two families
of methods: the algebraic methods, including the nuclear norm and 
factorization norm, and combinatorial methods, including discrepancy 
and the rectangle or corruption bound. The algebraic methods and discrepancy
give lower bounds on quantum communication
complexity, whereas the rectangle bound can show polynomial lower bounds 
on randomized communication complexity for problems known to have 
logarithmic quantum protocols.

A longstanding open problem is whether there are total functions
for which there is an exponential gap between classical and quantum
communication complexities.  
Many partial results have been given~\cite{NS96,BCWW01,BJK,KKRW}, 
most recently~\cite{KR11}.
\enlever{
In the simultaneous messages model without shared entanglement, {\sc Equality} 
requires
$\Omega(\sqrt{n})$ bits  in the randomized model~\cite{NS96}, whereas there is a 
$O(\log(n))$ quantum protocol using fingerprinting~\cite{BCWW01}.
Bar-Yossef \textit{et al.} showed that for the Hidden Matching problem, a relational
problem, there is an $O\log(n))$ quantum protocol and an $\Omega(\sqrt{n})$ randomized 
lower bound~\cite{BJK}. Gavinsky \textit{et al.} gave a similar result for a
partial function based on the Hidden Matching problem~\cite{KKRW}.  
Very recently, Klartag and Regev gave a separation for a stronger
model, but still for a promise problem.
They showed that for the Vector Subspace problem, there is a one-way
logarithmic protocol but the randomized two-way lower bound is 
$\Omega(n^{1/3})$~\cite{KR11}.  }
These strong randomized lower bounds all use the distributional
model, in which the randomness of the protocol is replaced
by randomness in the choice of inputs, which are sampled according to
some hard distribution.  
The equivalence of the randomized and distributional models, 
due to Yao's minmax theorem~\cite{yao83}, comes from strong duality
of linear programming.  This technique appears
to be inherently non-applicable to quantum communication
complexity (see for instance \cite{GW} which considers a similar
question in the setting of query complexity),
and the rectangle bound, as a result, was understood to be a inapplicable
to quantum lower bounds.

Contrary to previous combinatorial type lower bounds, the partition
bound is proven directly for randomized protocols, without first
going to the distributional model.
Although the partition bound re-introduces linear programming duality,
the dual variables 
can no longer be interpreted as a (hard) distribution on the inputs.
By the same token, it is 
harder to get intuition on how to obtain concrete lower bounds
for explicit functions.

\subsection{Bell experiments} 

Quantum nonlocality gives us a different viewpoint  from
which to consider lower bounds for communication
complexity.  A fundamental question of quantum mechanics
is to establish experimentally whether nature is truly \emph{nonlocal}, as predicted by quantum mechanics, or whether
there is a purely classical (i.e., \emph{local}) explanation to the phenomena that
have been predicted by quantum theory and observed in the lab.  
In an experimental setting, two players share an entangled state
and each player is given a measurement to perform.  The outcomes
of the measurements are predicted by quantum mechanics and
follow some probability distribution $p(a,b|x,y)$, where $a$ is the
outcome of Alice's measurement $x$, and $b$ is the outcome of Bob's
measurement $y$. (We write $\bp$ for the distribution,
and $p(a,b|x,y)$ for the individual probabilities.) A Bell test~\cite{Bell64} consists of estimating all the
probabilities $p(a,b|x,y)$ and computing a Bell functional, 
or linear function, on these values.  The Bell functional $B(\bp)$
is chosen together with a threshold $\tau$  so that any local classical
distribution $\bp'$ verifies $B(\bp')\leq \tau$, but the chosen
distribution $\bp$ violates this inequality: $B(\bp) > \tau$.

Although there have been numerous
experiments that have validated the predictions of quantum mechanics, 
none so far has been totally ``loophole-free".  A loophole can be introduced, for
instance, when the state preparation and the measurements are
imperfect, or when the detectors are partially inefficient so that 
no measurement is registered
in some runs of the experiment, or if the entangled particles are
so close that communication may have taken place in the course
of a run of the experiment.  In such cases, there are classical explanations for 
the results of the experiment.  For instance, if the detectors
were somehow coordinating their behavior, they may choose to
discard a run, and though the conditional probability
(conditioned on the run not having been discarded) may look quantum,
the unconditional probability may very well be classical.
This is called the detection loophole.  When an experiment aborts with probability at most
$1-\eta$, we say that the efficiency is $\eta$.  (Here we assume
that individual runs are independent of one another.)
To close the detection loophole, the efficiency has to be high enough
so that the classical explanations are ruled out.
Gisin and Gisin show for example  that the EPR correlations can be reproduced 
classically with 75\% detector efficiency~\cite{GG99}.
However, in practice, whenever the
detectors can be placed far apart enough to prevent 
communication from taking place (typically in optics setups), 
the efficiency is extremely small (on the order of 10\%), which
is far too small to close the detection loophole.

What can Bell tests tell us about communication complexity?
Both are measures of how far a distribution is from the set of
local distributions (those requiring no communication), and one would
expect that if a Bell test shows a large violation for a
distribution, simulating this distribution should require a lot 
of communication,
and vice versa. Degorre \textit{et al.} showed that the factorization norm
amounted to finding large Bell inequality violations for
a particular class of Bell inequalities~\cite{DKLR11}.  Here, we introduce 
a new class
of Bell inequalities whose violation
corresponds to the partition bound.

\subsection{Summary of results} 

If we assume there is a $c$-bit classical communication protocol
where Alice and Bob output $a,b$ with distribution $p(a,b|x,y)$ when Alice's
input is $x$ and Bob's input is $y$,  then there is a protocol
without communication that outputs according to $\bp$ (conditioned
on the run not being discarded) that uses shared randomness and
whose efficiency is $2^{-c}$: both players guess a transcript,
and if they disagree with the transcript, they abort.  Otherwise
they follow the protocol using the transcript.  
As others have observed~\cite{M01, BHMR03},
one can immediately derive a lower bound: let $\eta$ be the
maximum efficiency of a protocol without communication 
that successfully simulates $\bp$ with shared randomness.  
We define  $\eff(\bp) = 1/\eta$,
and $\log(\eff(\bp))$ is a lower bound on the communication complexity of simulating $\bp$.  
This gives a surprisingly strong bound.  We show that 
it coincides with the partition bound (in the special case of computing functions).

When we turn to the dual formulation, we get 
a natural physical interpretation, that of Bell inequalities.
To prove a lower bound amounts to finding a good Bell inequality
and proving a large violation.
This is similar to finding a hard distribution and proving a lower
bound in the distributional model of communication; but it is
much stronger since the Bell functional is not required to 
have positive coefficients that sum to one.

Our approach leads naturally to a new ``quantum partition
bound" which is a strong lower bound on quantum communication complexity.
Let $\eff^*(\bp)= 1/\eta^*$, where  $\eta^*$ is the
maximum efficiency of a protocol without communication 
that successfully simulates~$\bp$ with shared entanglement. 
In the one-way setting, we show that the quantum partition bound is tight.

Allowing for 
runs to be discarded with some probability 
has been studied in different models of computation 
such as post-selection, and zero-error (Las Vegas) randomized computation.
(Jain and Klauck~\cite{JK10} in fact introduce a Las Vegas partition bound for 
zero-error protocols.)
This is a stronger requirement than allowing a probability of error since
the errors must be flagged. 
Lee and Shraibman give a proof of the factorization norm ($\gamma_2$) lower bound
on (quantum) communication complexity based on the best bias one can
achieve with no communication~\cite[Theorem~60]{LS09b} 
(attributed to Buhrman; see also Degorre \textit{et al.}~\cite{DKLR11}).
In light of our formulation of the (quantum) partition bound,
it is an easy consequence that the (quantum) partition bound is an upper bound on
$\gamma_2$ (see e.g.~\cite{LS09} for definitions of the factorization norm $\gamma_2$ and the related nuclear norm $\nu$, as well as~\cite{DKLR11} for their extensions to the communication complexity of distributions),
making it the strongest known bound on quantum communication complexity to date.

The following gives a summary of our results. 
Full definitions and statements are given in the main text.
Let $\prt(\bp)$ be the partition bound for a distribution~$\bp$ 
(defined in Section~\ref{sec:prt}).
$R_0(\bp)$ denotes
the communication complexity of simulating~$\bp$ exactly using
shared randomness and classical communication,
and $Q_0^*(\bp)$ denotes
the communication complexity of simulating $\bp$ exactly using
shared entanglement and quantum communication. 
One-way communication, where only Alice sends a message to Bob, is denoted by the superscript~$\rightarrow$. In the simultaneous messages model, each player sends a message to the referee, who does not know the inputs of either player, and has to produce the output.  This is denoted by the superscript~$\parallel$.
Shared entanglement is indicated by the superscript $*$.
For any distribution $\bp$, 
\begin{myitemize}
\item Theorem~\ref{thm:equiv-eff-prt}: $\prt(\bp) = \eff(\bp)$,
\item Theorem~\ref{thm:lower-eff-quantum}: $Q_0^*(\bp) \geq \frac{1}{2} \log(\eff^*(\bp))$,
\item Theorem~\ref{thm:gamma-upper}: $\gamma_2(\bp) \leq 2\ \eff^*(\bp)$ and $\nu(\bp) \leq 2\  \eff(\bp)$ (for nonsignaling $\bp$),
\item Theorem~\ref{thm:upper-two-way}: $R^{*,\parallel}_\epsilon(\bp) \leq O(\eff^*(\bp))$ and $R^{*}_\epsilon(\bp) \leq O(\sqrt{\eff^*(\bp)})$.
\end{myitemize}

In the case of one-way communication, the upper bounds are much tighter.
The one-sided efficiency measure, which we denote
$\eff^\rightarrow$ is given in Definition~\ref{defn:one-way-eff}.   
\begin{myitemize}
\item Theorem~\ref{thm:lower-one-way} 
: $Q_0^{*,\rightarrow} (\bp)\geq\frac{1}{2}\log(\eff^{*,\rightarrow}(\bp))$ and
\item Theorem~\ref{thm:upper-one-way}:  $Q_\epsilon^{*,\rightarrow}(\bp) \leq \frac{1}{2}\log(\eff^{*,\rightarrow}(\bp))+O(1)$.

\end{myitemize}

\enlever{
We can also study the amount of communication that is required
to simulate a distribution $\bp$, if some fixed inefficiency
$1-\eta$ is also allowed.  Our lower bound in this model
is denoted $R_0^\eta(\bp)$.

In the case where some inefficiency is allowed, we introduce a quantity $\eff^\eta(\bp)$ in order to obtain a lower bound on $R_0^\eta(\bp)$. Another interesting quantity is $\effl(\bp)=\eta^{-1}$, where $\eta$ is the maximal efficiency such that $R_0^\eta(\bp)=0$.
Here we can show:
\begin{myitemize}
\item  Lemmas~\ref{lem:upper-eta} and~\ref{lem:tradeoff}: $\log(\eff^\eta(\bp))\leq R^\eta_0(\bp)$ and $\eta\cdot\effl(\bp)\leq\eff^\eta(\bp)\leq\eta\cdot\eff(\bp)$.
\end{myitemize}
}

We can use smoothing to handle $\epsilon$ error, and we can formulate the bounds
to allow both $\epsilon$ error and $\eta$ efficiency.
(In the latter case, for boolean functions, this is equivalent to relaxing the exactness constraints in the linear programs.)
For simplicity we have omitted these details in this summary.

We prove strong Bell violations using these new techniques for two problems studied in~\cite{BRSW11}.
The first is based on the Hidden Matching problem, from which we derive a distribution that can
be simulated with zero communication and an  $n$-dimensional shared quantum state.  For 
this distribution, we prove an exponential Bell inequality violation for one-way efficiency resistant Bell inequalities, where one player is allowed to abort. 
In contrast, Junge et al~\cite{JPP+10,JP11} have shown
a linear upper bound on normalized Bell inequalities, as a function of the dimension of the shared state. 
Therefore, our lower bound exhibits an exponential gap between 
the usual normalized Bell inequalities and the 
new (one-way) efficiency resistant Bell inequalities.
\enlever{We also study a distribution related to the Khot-Vishnoi game and prove a  logarithmic communication lower bound for this distribution which can be simulated exactly using zero communication and shared entanglement.}

\subsection{Related work}
Massar exhibits a Bell inequality that is more robust against detector
inefficiency based on the
distributed Deutsch Josza game~\cite{M01}.
The Bell inequality is derived from the lower bound on communication 
complexity for this promise problem \cite{BCW98,BCT99}.
Massar shows an upper bound of $\eff(\bp)$ on expected
communication complexity of simulating $\bp$.
He also states, but does not claim to prove, that a
lower bound can be obtained as the logarithm of the efficiency.
Buhrman \textit{et al.}~\cite{BHMR03, BHMR06} show how to get
Bell inequalities with better resistance to detector inefficiency
by considering multipartite scenarios where players share
GHZ type entangled states.  Their technique is based on
the rectangle bound and they derive a general tradeoff
between monochromatic rectangle size, efficiency, and communication.
They show a general lower bound on multiparty communication complexity
which is exactly as we describe above.

Buhrman \textit{et al.}~\cite{BRSW11} show gaps between quantum
and classical winning probability for games where the players are 
each given inputs and attempt, without communication, to produce 
outputs that satisfy some predicate.  In the classical case they use 
shared randomness and in the quantum case, they use shared entanglement.
Winning probabilities are linear  so these translate to large 
Bell inequality violations.

V\'ertesi~\textit{et al.} show that there is a distribution with boolean outputs 
$\bp\in\Q$, based on partially entangled states, such that (in our language) 
$\eff^\rightarrow (\bp)=\Omega(  2^n)$~\cite{VPB10}.  Therefore, our results imply that $R_0^\rightarrow(\bp)=n$.  Since the
states are nearly separable, however $R_\epsilon^\rightarrow(\bp)=0$ for large enough $\epsilon$.

Lower bounds for communication complexity of simulating distributions
were first studied in a systematic way by Degorre \textit{et al.}~\cite{DKLR11}.  These bounds are
shown to be closely related to the nuclear norm and factorization norm~\cite{LS09},
and the dual expressions are interpreted as Bell 
inequality violations.

Following up on  the results in this paper, Kerenidis et al~\cite{KLL+12} used the notion of efficiency of 
zero-communication protocols to show that the information cost is bounded below
by a relaxation of the partition bound which is larger than the smooth rectangle
and $\gamma_2$ bounds.  Jain and Yao \cite{JY12} followed up with a strong direct product theorem for the communication complexity of all functions for which an optimal lower bound can be shown using the smooth rectangle bound.  Using a similar notion of zero-communication protocols,
Gavinski and Lovett \cite{GL13} showed that the log rank conjecture is equivalent to an upper bound which is polylogarithmic in the rank on the zero-communication cost. The notion of zero-communication cost that they use is the non-constant efficiency (Definition~\ref{def:nc-eff} and Lemma~\ref{lem:nc-eff}.)
\section{Preliminaries}

\subsection{Classical partition bound}
The partition bound of Jain and Klauck~\cite{JK10} is given as a linear program,
following the approach introduced by Lov\'asz~\cite{L90} and
studied in more depth by Karchmer \textit{et al.}~\cite{KKN95}.
It differs from the rectangle and other combinatorial bounds
in that it is formulated directly on the randomized protocol,
as opposed to first applying Yao's minmax theorem to reduce to
a deterministic protocol with distributional inputs.
From a $c$-bit, $\epsilon$-correct randomized protocol one can infer a 
distribution over rectangle partitions of size at most $2^c$, 
where each rectangle is assigned 
an output value $z$.  Set weights $w_{R,z}$
to be the probability that  rectangle $R$
occurs with label $z$ (the same rectangle may occur in 
different partitions, with
different labels and different probabilities).  This is
a feasible solution to the following linear program.

\begin{defn}[Partition bound~\cite{JK10}]
Let $f:\X\times \Y \rightarrow \Z$ be any partial function
whose domain we write $f^{-1}$.  Then $prt_\epsilon(f)$
is defined to be the optimal value of the linear program,
where $R$ ranges over the rectangles from $\X\times \Y$ and
$z$ ranges over the set $\Z$:
\begin{align*}
\prt_\epsilon(f) = &\min_{w_{R,z}\geq 0}  && \sum_{R,z} w_{R,z}\\
	 &\text{subject to} 
	&&  \sum_{R: (x,y)\in R} w_{R, f(x,y)}\geq 1-\epsilon  &\quad \forall x,y \in f^{-1}\\
	& && \sum_z \sum_{R: (x,y)\in R} w_{R,z} =1  &\quad \forall x,y \in \X\times \Y\\ 
\end{align*}
\end{defn}
The feasible solution sketched above verifies all the
constraints and the objective value is  at most $ 2^c$, 
so $R_\epsilon(f)\geq \log(\prt_\epsilon(f))$.  
The partition bound subsumes
almost all previously known techniques~\cite{JK10}, in particular
the factorization norm~\cite{LS09}, smooth rectangle~\cite{JK10} and rectangle  or corruption bound~\cite{yao83}, and discrepancy~\cite{CG85,BNS89}.  
Bounds not known to be subsumed by the partition bound include 
the approximate rank~\cite{K96, BW01} and the information complexity~\cite{CSWY01}.


\subsection{Local and quantum distributions}
\label{sec:distributions}
Given a distribution
$\bp$, how much communication is required if Alice is given
$x\in \X$, Bob is given $y\in \Y$, and their goal is to output $a,b\in \A\times\B$
with probability $p(a,b|x,y)$?

Some classes of distributions are of interest and have been
widely studied in quantum information theory since the seminal paper of Bell~\cite{Bell64}.  
The local
deterministic distributions, denoted $\bl\in \L_{\det}$, are the ones where
Alice outputs according to a deterministic strategy, i.e., a (deterministic)
function of $x$, and Bob independently outputs as a function of $y$,
without exchanging any communication.
The local distributions $\L$ are any distribution over the local
deterministic strategies. Mathematically this corresponds to taking convex
combinations of the local deterministic distributions, and
operationally to zero-communication protocols with shared randomness.

We  focus our attention in this paper on local strategies that are allowed
to abort the protocol with some probability.  When they abort, they output the
symbol $\bot$.
We will use the notation $\L_{\det}^\bot$  and $\L^\bot$ to
denote these strategies, where $\bot$ is added to the possible outputs 
for both players, and $\bot\not\in \A\cup \B$.
Therefore, when~$\ell \in \L_{\det}^\bot$ or $\L^\bot$, $l(a,b|x,y)$ is {\em not} conditioned on 
$a,b\neq \bot$ since $\bot$ is a legal output for such distributions. 

The quantum distributions, denoted $\bq\in \Q$, are the ones that result from
applying measurements to each part of a shared entangled
bipartite state.  Each player outputs the  measurement
outcome.  In communication complexity terms, these are
zero-communication protocols with shared entanglement.
If the players are allowed to abort, then the corresponding
set of distributions is denoted $\Q^\bot$.

Boolean (and other) functions can be cast as a sampling problem as follows.
Consider a boolean function~$f:\X\times \Y \rightarrow \{0,1\}$ 
whose communication complexity we wish to study (non-boolean functions and relations can
be handled similarly).
First, we split the output so that if $f(x,y)=0$, Alice
and Bob are required to output the same bit, and
if $f(x,y)=1$, they output different bits.
Let us further require Alice's marginal distribution to be
uniform, likewise for Bob, so that the distribution is well defined.  
Call the resulting distribution
$\bp_f$.  If $\bp_f$ were local, $f$
could be computed with one bit of communication 
using shared randomness: Alice sends her output to Bob, 
and Bob XORs it with his output.
If $\bp_f$ were quantum, there would be a 1-bit
protocol with shared entanglement for $f$.  We
are usually interested in distributions 
requiring nontrivial communication complexity, and
lie well beyond these sets.

\enlever{Finally, a distribution is nonsignaling if for each player, its marginal output 
distributions given by $p(a|x,y)=\sum_b p(a,b|x,y)$, for Alice, an
$p(b|x,y)=\sum_a p(a,b|x,y)$  do not depend on the other player's input.
When this is the case we write the marginals as $p(a|x)$ and $p(b|y)$.
For any Boolean function $f$, the distribution $\bp_f$ described is 
nonsignaling since the marginals are uniform.
}

\subsection{Communication complexity measures}
We use the following notation for communication complexity of distributions.
$R_\epsilon(\bp)$ is the
minimum amount of communication necessary to reproduce
the distribution $\bp$ in the worst case,
up to $\epsilon $ in total variation distance for all $x,y$. 
We write $|\bp-\bp'|_1\leq \epsilon$ to mean that
for any $x,y$, $\sum_{a,b} | p(a,b|x,y) - p'(a,b|x,y)|\leq \epsilon$.

$R_0^\eta(\bp)$ is the amount of communication 
needed to reproduce $\bp$ exactly with a protocol
which may abort with probability at most $1-\eta$
for any input $x,y$ (the probability that it aborts may depend on $x,y$). 
The probability produced
by the protocol is conditioned
on the event that neither player aborts.
When the player aborts it outputs $\bot$.  

For quantum communication, we use $Q$ to denote quantum
communication, and we use the superscript~$*$ 
to denote the presence of shared entanglement. We use
superscripts $\rightarrow$ for one-way communication (i.e, when only Alice can send a message to Bob),
and $\parallel$ for simultaneous messages (i.e., when Alice and Bob cannot communicate with each other, but are only allowed to send a message to a third party who should produce the final output of the protocol).  
The usual relation $Q_\epsilon^{\eta,*}(\bp) \leq R_\epsilon^\eta(\bp)$ holds
for any $\epsilon,\eta, \bp$.
Moreover, since one can always output at random instead of aborting,
which introduces at most $1-\eta$ error for each $x,y$, we have the following relation between $R_\epsilon(\bp)$ and $R_\epsilon^\eta(\bp)$.
\begin{lemma}\label{lem:eta-epsilon}
For any $\epsilon,\eta$ and any distribution $\bp$, we have
 $R_{\epsilon+(1-\eta)}(\bp)\leq R_{\epsilon}^{\eta}(\bp)$.
\end{lemma}

For all the models of randomized communication,
we assume shared randomness between the
players.  Except in the case of simultaneous messages,
this is the same as private randomness up to a logarithmic
additive term~\cite{newman91}. (Ref.~\cite{DKLR11} sketches 
a proof of how to adapt Newman's theorem to
the case of simulating distributions.)

\section{Partition bound and detector inefficiency}

\subsection{The partition bound for distributions}
\label{sec:prt}

We extend  the partition bound to the
more general setting of simulating a distribution $p(a,b|x,y)$
instead of computing a function.
Protocols with shared randomness and communication
also lead to a distribution over rectangle partitions;
however, since each player ouputs a value,
the label associated with each rectangle is
a local deterministic distribution, denoted by
$\bl$.   The following definition applies to
protocols that use communication
and allow the players to abort a run with some
fixed probability $1-\eta$.  The partition bound corresponds to
the case $\eta=1$ and the Las Vegas partition bound~\cite{JK10} is closely
related to the case $\eta=1/2$.

\begin{defn}\label{defn:prt}  For any distribution $\bp=p(a,b|x,y)$, over inputs
$x\in \X, y\in \Y$ and outputs $a\in \A,b\in \B$, 
define $\prt^{\eta}(\bp)$ to be the optimal value of the following linear
program.  The variables of the program are $\eta_{x,y}$ and $ w_{R,\bl} $,
where $R$ ranges over the rectangles from $\X\times \Y$ and
$\bl$ ranges over the local deterministic distributions
with inputs in $R$ and with outputs in $\A\times \B$.
\begin{align*}
 \prt^\eta(\bp) = 
         &\min_{w_{R,\bl\geq 0 },\eta_{x,y}}    
         &&\sum_{R,\bl\in\L_{\det}} w_{R,\bl} 
          \\
         &\text{subject to } 
         &&\sum_{R,\bl\in\L_{\det}: x,y\in R} w_{R,\bl} \cdot l(a,b|x,y) 
				= p(a,b|x,y) \cdot \eta_{x,y} 
        & \quad\forall x,y,a,b\in \X{\times}\Y{\times}\A {\times}\B\\
        & && \eta \leq \eta_{x,y} \leq 1 
        &  \quad  \forall x,y \in \X\times \Y.
\end{align*}

When $\eta=1$ we write $\prt(\bp)=\prt^1(\bp)$, and the linear program
simplifies:
\begin{align*}
 \prt(\bp)  = 
        &\min_{w_{R,\bl}\geq 0 }  
        &&   \sum_{R,\bl} w_{R,\bl}   \\
        &\text{subject to } 
        && \sum_{R,\bl: x,y\in R} w_{R,\bl} \cdot l(a,b|x,y) = p(a,b|x,y)  
        & \quad \forall x,y,a,b\in \X{\times}\Y{\times}\A {\times}\B\\
\end{align*}

For randomized communication with error,
$\prt^\eta_\epsilon(\bp)=\min_{|p'-p|_1\leq \epsilon} \prt^\eta(\bp')$.

\end{defn}

\begin{thm}\label{thm:partition}  For any distribution $\bp$,
$R_\epsilon^\eta(\bp) \geq \log(\prt_\epsilon^\eta(\bp))$.  
\end{thm}

We have included a direct proof of the theorem in Appendix~\ref{apx:pf},
which for $\eta=1$ is essentially the same as the original partition bound,
sketched above, where output values $z$ are replaced with local
deterministic strategies $\bl$.  We will now turn to an alternative, 
arguably simpler, proof by
introducing the efficiency bound.

\subsection{The efficiency bound}

For any distribution $\bp$, $\eff(\bp)$ is the inverse of the maximum efficiency
sufficient to simulate it classically with shared randomness,  without communication.
\begin{defn}
For any distribution $\bp$ with inputs $\X\times \Y$ and outputs in $\A\times \B$,
$\eff(\bp)= 1/\zeta_\opt$, where $\zeta_\opt$ is the optimal value of the following
linear program.  The variables are $\zeta$ and $q_\bl$, where 
$\bl$ ranges over local deterministic protocols with inputs taken from $\X\times\Y$
and outputs in $\A\cup\{\bot\}\times \B\cup\{\bot\}$.
\begin{align*}
   \zeta_\opt = 
    & \max_{\zeta,q_\bl\geq 0} 
    &&  \zeta \\ 
    &\text{subject to } 
    &&\sum_{\bl\in\L_{\det}^\bot} q_\bl l(a,b|x,y)=\zeta p(a,b|x,y)
    & \quad \forall x,y,a,b\in \X{\times}\Y{\times}\A {\times}\B\\
    & &&\sum_{\bl\in\L_{\det}^\bot} q_\bl=1
\end{align*}
For randomized communication with error,
define $\eff_\epsilon(\bp)=\min_{|p'-p|_1\leq \epsilon} \eff(\bp')$.
\end{defn}
The first constraint says that the 
local distribution, conditioned on both outputs differing from $\bot$,
equals~$\bp$, and the second is a normalization constraint.
Note that the efficiency $\zeta$ is the same for every input
$x,y$.  
This is surprisingly important and 
the relaxation $\zeta_{x,y}\geq \zeta$ 
does not appear to coincide with the partition bound. 
Other more realistic variants (for the Bell setting), 
such as players aborting independently
of one another, could be considered as well. (We note that
this would not result in a linear program.)

\begin{thm}~\cite{M01, BHMR03} \label{thm:lower-eff}
 $R_\epsilon(\bp)\geq \log\eff_\epsilon(\bp).$
\end{thm}
\begin{proof}[Proof]
 Let $P$ be a randomized communication protocol for a distribution $\bp'$ with $|\bp-\bp'|_1
\leq \epsilon$, using $t$ bits of communication. $P$ is a convex combination of deterministic protocols: we denote by $\Lambda$ the source of public randomness and by $P_\lambda$ the deterministic protocol corresponding to $\lambda \in \Lambda$. We assume that the total number of bits exchanged is independent of the execution of the protocol, introducing dummy bits at the end of the protocol
if necessary.
Let~$q_l$ be the following distribution over local deterministic protocols $\bl$:
using shared randomness Alice and Bob first pick some random value $\lambda$ according to $\Lambda$
and then they pick a transcript $T\in\{0,1\}^t$.
If $T$ is consistent with $P_\lambda$, Alice outputs according to $P_\lambda$, otherwise she outputs~$\bot$; similarly for Bob. 
We claim that for each $\lambda \in \Lambda$ only one transcript is valid for Alice and Bob simultaneously, so the probability that neither player outputs~$\bot$ is exactly $2^{-t}$. Indeed, define $A_x = \{T: \exists y', T=T_{x,y'}\}$ and $B_y = \{T: \exists x', T=T_{x',y}\}$.  Then $A_x \cap B_y = \{T_{x,y}\}$ because $T$ in $A_x \cap B_y$ means that there exist $x'$ and $y'$ such that $T=T_{x',y}=T_{x,y'}$. By the rectangle property of the transcripts of $P_\lambda$ it must be the case that $T=T_{x,y}$. Furthermore, if we condition on not aborting, this protocol does exactly the same thing as $P$. This distribution therefore satisfies the constraints of $\eff(\bp')$ with $\zeta=2^{-t}$.
\end{proof}


\begin{thm}\label{thm:equiv-eff-prt}
For any distribution $\bp$,
$\eff(\bp)=\prt(\bp)$.
\end{thm}

\begin{proof}
 In the partition bound, a pair $(\bl,R)$, where $\bl$ is a local distribution with outputs in $\A\times\B$ and $R$ is a rectangle, determines a local distribution $\bl_R$ with outputs in $(\A\cup\{\bot\})\times(\B\cup\{\bot\})$, where Alice outputs as in~$\bl$ if $x\in R$, and outputs $\bot$ otherwise (similarly for Bob). 
 Let $(a_0,b_0) \in \A\times\B$ be an arbitrary pair of outputs. In the efficiency bound, a distribution $\bl \in \L_{\det}^{\bot}$ defines both a rectangle being the set of inputs where neither Alice nor Bob abort, and a local distribution $\bl' \in \L_{\det}$ where Alice outputs as $\bl$ if the output is different from $\bot$ and $a_0$ otherwise (similarly for Bob with $b_0$).
We can transform the linear program for $\prt(\bp)$ into the linear program for $\eff(\bp)$ by making the change of variables:
$ \zeta=\left(\sum_{R,\bl}w_{R,\bl}\right)^{-1}$ and 
 $q_{\bl_R}=\zeta\ w_{R,\bl}.$
\end{proof}

We define  $\eff^\eta(\bp)$ which is equal to $\prt^\eta(\bp)$.
The details are given in Appendix~\ref{app:eff-eta}.

\subsection{Lower bound for quantum communication complexity}

By replacing local distributions by quantum distributions
we get 
a strong new lower bound on quantum communication
that subsumes the factorization norm.
Insomuch as the partition bound is an extension of the rectangle bound, this quantum analogue of the partition bound can be thought
of as 
a quantum extension of the rectangle bound, 
circumventing Yao's minmax theorem.

\begin{defn}
For any distribution $\bp$ with inputs $\X\times \Y$ and outputs $\A\times \B$,
$\eff^*(\bp)= 1/\eta^*$, with $\eta^*$ the optimal value of the following
(non-linear) program.
\begin{align*}
   \max_{\zeta,\bq\in \Q^\bot} &&& \zeta \\
   \text{subject to }          &&& q(a,b|x,y)=\zeta p(a,b|x,y)
		 &\forall x,y,a,b\in \X{\times}\Y{\times}\A {\times}\B\\
\end{align*}
As before, we let $\eff_\epsilon^*(\bp)= \min_{|p'-p|_1\leq \epsilon} \eff^*(\bp')$.
\end{defn}

\begin{thm}\label{thm:lower-eff-quantum}
 $Q^*_\epsilon(\bp)\geq \frac{1}{2}\log\eff_\epsilon^*(\bp).$
\end{thm}

\begin{proof}
Let $Q$ be a $t$ qubit communication protocol for $\bp'$ with $|\bp'-\bp|_1\leq \epsilon$.
We use teleportation and classical communication to send every qubit, hence obtaining an entanglement-assisted protocol using at most~$2t$ bits of classical communication. We introduce dummy bits so that the number of bits exchanged is exactly~$2t$, independently of the execution of the protocol.
Then proceed as before: guess the a uniform classical transcript using the 
shared randomness.
Then Alice and Bob will simulate the previous entangled-assisted communication protocol (performing the measurements) checking that the communication in the transcript is consistent with their execution of the protocol. If it is the case, they output according to the protocol, and they abort otherwise.
If we fix the outputs of the measurements then the protocol is deterministic so, as before there is exactly one transcript which is valid simultaneously for $x$ and $y$, and the efficiency is $2^{-2t}$.
Conditioning on not aborting, this protocol outputs exactly with the same distribution
as before. The result is a protocol using zero communication and
entanglement with efficiency $2^{-2t}$, satisfying
the constraints.
\end{proof}

Since the local distributions
form a subset of the quantum distributions,
$\eff^*(\bp) \leq \eff(\bp)$ for any~$\bp$.

\subsection{One-way efficiency bound}

Because of its rectangle-based formulation, the partition and efficiency bounds can easily be
tailored to the case of one-way communication protocols.
In the case of the partition bound, we consider only rectangles 
of the form $X\times Y$ with $Y=\Y$.  In the case
of the efficiency bound, this amounts to only letting
Alice abort the protocol. 
 The set of local (resp.~quantum) distributions
where only Alice can abort is denoted $\L_{\det}^{\botA}$
(resp.~$\Q^{\botA}$).
\begin{defn}\label{defn:one-way-eff}
Define $\eff^{\rightarrow}$ and $\eff^{*,\rightarrow}$ as 
\begin{align*}
\enlever{
(\eff^{\ind,\rightarrow}(\bp))^{-1}&=\max_{\zeta}\zeta &\textrm{subject to }& \sum_\bq q_\ql q(a,b|x,y)=\zeta p(a,b|x,y)& \\
&&&\forall a\in \A,b\in \B,x,y\in \X\times\Y\\
&&&q_\bl q(b|y)=p(b|y)&\forall b\in B,y\in \Y\\
}
(\eff^{\rightarrow}(\bp))^{-1}=
     &\max_{\zeta,q_\bl\geq 0}
     &&\zeta \\
     &\textrm{subject to }
     && \sum_{\bl\in\L_{\det}^{\botA}} q_\bl l(a,b|x,y)=\zeta p(a,b|x,y)
     &\quad \forall a\in \A, b\in B,x,y\in \X\times \Y\\
     & &&\sum_{\bl\in\L_{\det}^{\botA}}q_\bl=1\\
(\eff^{*,\rightarrow}(\bp))^{-1}=
     &\max_{\zeta,\bq\in \Q^{\botA}}
     &&\zeta \\
     &\textrm{subject to }
     && q(a,b|x,y)=\zeta p(a,b|x,y)
     & \quad \forall a\in \A, b\in B,x,y\in \X\times \Y.\\
\end{align*}
\end{defn}

\enlever{
We also have
\begin{lemma}
 $\eff^{\ind,\rightarrow}(\bp)\geq\eff^{\rightarrow}(\bp)$.
\end{lemma}
}

\begin{thm}\label{thm:lower-one-way}
  $R_0^{\rightarrow}(\bp)\geq \log\eff^{\rightarrow}(\bp)$ and
  $Q_0^{*,\rightarrow}(\bp)\geq \frac{1}{2}\log\eff^{*,\rightarrow}(\bp)$.
\end{thm}
The proof is similar to the two-way case.

\subsection{Efficiency is larger than $\gamma_2$}
Jain and Klauck show that the partition bound is an upper 
bound on $\gamma_2$ for boolean functions (in fact they show that
the weaker smooth rectangle bound is an upper bound on $\gamma_2$)~\cite{JK10}.
The lower bounds $\nu$ and~$\gamma_2$
were extended to nonsignaling distributions by Degorre \textit{et al.}~\cite{DKLR11}.  
Nonsignaling distributions, including quantum distributions, 
have marginal distributions independent of
the other player's input. 

\begin{defn}[Non-signaling distributions]
\label{defn:causal}
A distribution $\mathbf p$ is nonsignaling if
$ \forall a,x,y,y', \sum_b p(a,b|x,y) = \sum_b p(a,b|x,y')$,
and $ 
	\forall b,x, x',y, \sum_a p(a,b|x,y) = \sum_a p(a,b|x',y).$
\end{defn}

\begin{defn}[\cite{DKLR11}]\label{defn:gamma2}
For any nonsignaling distribution $\bp$,
\begin{myenumerate}
\item $\nu(\mathbf{p}) =
\min \{ \sum_i  \abs{q_i} :
	\exists \mathbf{p}_i\in\L,q_i\in \Real, \mathbf{p} = \sum_i q_i \mathbf{p}_i \} $,
\item $\gamma_2(\mathbf{p}) =
\min \{ \sum_i\abs{q_i} :
	\exists \mathbf{p}_i\in \Q, q_i\in \Real, \mathbf{p} = \sum_i q_i \mathbf{p}_i \} $,
\end{myenumerate}
\end{defn}
Recall from Section~\ref{sec:distributions} the definition of the distribution $\bp_f$ for any
boolean function $f$.
It was shown that for any Boolean function $f$, the factorization norm 
$\gamma_2(f) = \Theta(\gamma_2(\mathbf{p}_f))$, and similarly for the
nuclear norm, $\nu(f)= \Theta(\gamma_2(\mathbf{p}_f))$~\cite{DKLR11}.

\begin{thm}
\label{thm:gamma-upper}
For any nonsignaling $\bp$,
$\nu(\bp)\leq 2\ \eff(\bp)$ and $\gamma_2(\bp)\leq 2\ \eff^*(\bp)$.
\end{thm}

\begin{proof}
We sketch the proof for $\gamma_2$ vs $\eff^*$. 
The proof for $\nu$ vs. $\eff$ is similar.

Let $\zeta,\bq$ be an optimal solution for $\eff^*(\bp)$.
Then $ q(a,b|x,y)=\zeta p(a,b|x,y),$  $\forall x,y,a,b\in \X{\times}\Y{\times}\A {\times}\B$,
where $\bq$ outputs $\bot$ with probability $1-\zeta$ for every $x,y$.
Define $\tilde{\bq}\in \Q$ as 
the distribution where the players output according to $\bq$ unless their outcome is $\bot$, in
which case they output independently from the other player, uniformly at random from $\A$ 
or $\B$.
Let
$\mathbf{r}$ be the distribution $\tilde{\bq}$ conditioned on one of the players having output $\bot$ when they ran $\bq$.  
We can write $\tilde{q}(a,b|x,y)= q(a,b|x,y) + (1-\zeta)r(a,b|x,y)$. 
Notice that $\mathbf{r}\in\Q$.
Therefore, on $\A\times\B$,
$\bp = \frac{1}{\zeta}\bq = \frac{1}{\zeta} \tilde{\bq} - \frac{1-\zeta}{\zeta} \mathbf{r}$ .
This is an affine combination of quantum distributions (that do not output $\bot$)
so $\gamma_2(\bp) \leq |\frac{1}{\zeta}| + |-\frac{1-\zeta}{\zeta}|= 2\eff^*(\bp)-1 $.
\end{proof}

For Boolean functions, the gap between $\nu$ and $\gamma_2$
is known to be at most a multiplicative constant (by 
Grothendieck's inequality).
However, there is no immediate way to conclude
similarly for $\eff$ vs.~$\eff^*$.
Since these are stronger bounds, determining the 
largest possible gap between these measures could lead 
to further evidence towards the existence, or not,
of exponential gaps between quantum and classical
communication complexity for total boolean functions.

\section{Detector resistant Bell inequalities}

\subsection{Dual formulation of the efficiency bound}

Proving lower bounds in the standard formulation of $\eff$ is 
difficult since it is formulated as the multiplicative inverse of a maximisation
problem, which translates to a universal quantifier on all the variables of
the optimization problem.
To prove concrete lower bounds, we will 
use the dual formulation, expressed as a maximisation
problem, where it suffices to
check a feasible solution against all of the constraints.

\begin{lemma}[Dual formulation of the efficiency bounds]
\label{lemma:eff-dual}
For any distribution $ \bp$, 
 \begin{align*}
 \eff(\bp)=&\max_{B} && B(\bp)\\
&\textrm{subject to } &&
 B(\bl)\leq 1 & \quad \forall \bl\in\L_{\det}^\bot \\
 \eff^*(\bp)=&\max_{B}&&  B(\bp)\\
&\textrm{subject to } &&
 B(\bq)\leq 1 & \quad \forall \bq\in\Q^\bot
\end{align*}
where $B$ ranges over all real linear functionals, with coefficients $B_{abxy}$
and  for any distribution $\bp$,
$B(\bp) =  \sum_{(a,b,x,y)\in \A\times \B\times \X\times \Y} B_{abxy} p(a,b|x,y)$.
(There are no coefficients for the abort outcomes.)
\end{lemma}

The first part uses linear programming duality and the second can
be shown using Lagrange multipliers.

These expressions are clearly reminiscent of Bell inequalities,
except for the introduction into the constraints of local strategies that may abort.
Let us compare the classes of Bell inequalities above, with those stemming from the lower bound
$\nu$, introduced in~\cite{DKLR11} and studied in Junge \textit{et al.}~and Buhrman \textit{et al.}~\cite{JPP+10,BRSW11}.
The dual of $\nu$ can be formulated as follows:

\begin{proposition}[\cite{DKLR11}]
\begin{align*}
 \nu(\bp)=&\max_{B}&& B(\bp)\\
&\textrm{subject to } &&
 \abs{B(\bl)}\leq 1 & \quad \forall \bl\in\L_{\det}
\end{align*}
\end{proposition}
We will call the family of Bell inequalities satisfying the constraints of 
the dual formulation  of $\nu$ {\em normalized Bell inequalities},
and those satisfying the constraints of the dual expression for $\eff$, {\em inefficiency resistant Bell
inequalities}.
Theorems~\ref{thm:gamma-upper} and~\ref{thm:lower-eff} tell us that $\frac{1}{2}\nu(\bp)\leq \eff(\bp) \leq 2^{R(\bp)}.$  This is not immediately apparent by comparing the two dual expressions, since 
there are two differences 
between the two families of Bell inequalities. One is a relaxation of the normalized Bell
inequalities by removing  the absolute value in the constraints.  
This increases the value of $\eff$, and one might worry that this will lead to unbounded violations. 
(Of course this is ruled out by the upper bound on $\eff$.)
The second difference goes in the other direction:
the constraint on inefficiency resistant inequalities is required to hold for all local 
strategies including those that may abort. 
Notice that to get a value of $\eff$ that is larger than $\nu$, 
there have to be local strategies that abort
whose Bell value is (far) less than -1, otherwise one is just exhibiting
a normalized Bell inequality with added constraints.

Concretely, how does one go about finding a feasible solution to the dual? 
Consider a distribution $\bp$ for which we would like to find a lower bound.
We construct a Bell inequality $B(\bp)=\sum_{a,b,x,y} B_{abxy} p(a,b|x,y)$ so that $B(\bp)$ is large, and
$B(\bl)$ is small for every $\bl\in\L^\bot$.
The goal is to balance the
coefficients $B_{abxy}$ so that they correlate well with
the distribution $\bp$ and badly with local strategies.
For $B(\bl)$ to be
small for local strategies, we apply a small weight or
even a penalty (negative weight) when the local strategy is incorrect.
For $B(\bp)$ to be large, we assign a positive coefficient
when the outcome is correct, or if it should occur with high probability.  
Weights can be zero when the input is not
contributing to the hardness of the problem.

The dual of the one-way efficiency bound can also be interpreted as Bell inequality violations.
\enlever{
}
\begin{lemma}[Dual formulation for one-way efficiency]
 For any $\bp$, 
\begin{align*}
 \eff^{\rightarrow}(\bp)=&\max_{B_{abxy}} && B(\bp)\\
&\textrm{subject to } && 
 B(\ell) \leq 1  &\quad \forall\bl\in\L_{\det}^{\botA}\\
 \eff^{*,\rightarrow}(\bp)=&\max_{B}&&  B(\bp)\\
&\textrm{subject to } &&
 B(\bq)\leq 1 & \quad \forall \bq\in\Q^\botA
\end{align*}
\end{lemma}
Notice that the constraint need only be verified for those local strategies where Alice is allowed to abort.

\subsection{Bell violation for the Hidden Matching distribution}

We show how to apply the efficiency bound to derive an exponential 
efficiency-resistant Bell violation for 
the Hidden Matching problem~\cite{BJK,KKRW,BRSW11}.
The Hidden Matching problem can be formulated as a game that can be won with probability 1 by a quantum
protocol with zero communication and an $O(n)$-dimensional 
shared quantum state.  
Buhrman \textit{et al.} show a normalized Bell violation of $\Omega(\frac{\sqrt{n}}{\log(n)})$ for this game. 
Our exponential violation 
should also be compared to  the results of Junge  \textit{et al.}~\cite{JPP+10} who  show 
that if a dostribution $\bp$ can be simulated with a quantum protocol with an
$n$ dimensional shared quantum state, then $\nu(p) \leq n$.
As we have discussed above, we are dealing here with a different family of Bell inequalities.  
More precisely,  we establish that there can be an exponential gap between
$\nu$ and $\eff^{\rightarrow}$.


We apply the partition bound method on the Hidden Matching probability distribution that 
we define here.  The Hidden Matching distribution is based on the Hidden Matching
problem of~\cite{BJK} adapted to the setting of games by Buhrman \textit{et al.}~\cite{BRSW11}.
We use many of the ideas and techniques from the latter to give the efficiency bound,
but some added tricks are needed to derive
the exponentially larger Bell violations.

\begin{defn}[Hidden Matching distribution]
Alice receives $x \in \{0,1\}^n$ and Bob receives a matching $M$ over vertices $\{1, \ldots, n \}$.
Alice has to output $a \in \{0,1\} ^ {\log(n)}$ and Bob has to output $d \in \{0,1\}$ and  $(i,j) \in M$ according to the following distribution, which we call the \emph{Hidden Matching distribution}:
$$\HM(a,d,i,j| x,M) =  \left\{
    \begin{array}{ll}
       \frac{2}{n^2} & \mbox{if } \langle a,i \oplus j\rangle \oplus d = x_i \oplus x_j \\
        0   & \mbox{otherwise.}
    \end{array}
\right.$$
\end{defn}

\begin{thm}[\cite{BRSW11}] $\HM\in \Q$, that is, $Q_0^*(\HM)=0$, and $\eff^*(\HM)=1$. The zero-communication quantum protocol for $\HM$ uses an $n$-dimensional shared quantum state. 
\end{thm}

\begin{thm} \label{thm:HM}
There exists a constant $\mathcal{C}>0$ such that for any $0\leq\epsilon<\frac{1}{2}$, 
$ \eff_{\epsilon}^{\rightarrow}(\HM) \geq  \frac{2^{\frac{\sqrt{n-1}}{2 \mathcal{C}}}} {n}(\frac{1}{2} -  \epsilon)$.
\end{thm}

From the one-way version of Lemma~\ref{lem:eta-epsilon},
 we obtain as a corollary:

\begin{cor}
There exists a constant $\mathcal{C}>0$ such that for any $0<\eta\leq 1$ and $0\leq\epsilon<\eta-\frac{1}{2}$, 
 $ R_{\epsilon}^{\eta,\rightarrow}(\HM) \geq \frac{\sqrt{n-1}}{2 \mathcal{C}}  - \log(n)+ \log(\eta-\epsilon-\frac{1}{2})$.
\end{cor}

We give a bound on $  \eff_{\epsilon}^{\rightarrow}(\HM) =\min_{\{p' : |p' - \HM|_1 \leq \epsilon\}}(\eff^{\rightarrow}(\bp'))$.  
To do that, we need to find a Bell functional which is bounded from above for any  local deterministic strategy where Alice may abort and which is large for any distribution close to the Hidden Matching distribution.
The main idea is to define a Bell functional which translates the bias of the winning probability of the strategy putting negative weights when the strategy fails and positive ones when it succeeds. Then we show that for any local deterministic strategy where Alice may abort, this bias is small, whereas it is big for any distribution close to the Hidden Matching distribution.
\enlever{To give an upper bound on the Bell value of any local deterministic strategy that may output $\bot$,
we will use the fact that such a strategy leads to a partition of Alice's inputs where she doesn't abort into rectangles. We will show an upper bound on the bias of each rectangle using the analysis from \cite{BRSW11}. However in this way the bound obtained depends on the size of this rectangle, and we will need to consider two different cases. If the rectangle is small enough, then we obtain a sufficiently good upper bound. If the rectangle is too big, we will need to subtract from the $B$s some constant that we will call $\mu$. Notice that in $\eff$, the constraint is that the Bell value of any local deterministic strategy is less than 1, and it is not  necessary the case for the absolute value as in $\nu$, that's why we can subtract without violating the constraint. The overall weight of those $\mu$ will not affect too much the Bell value of a distribution close to the Hidden Matching distribution which will remain large enough.}
Details of the proof are given in Appendix~\ref{app:HM}. 

\section{Upper bounds for one- and two-way communication
}

The efficiency bound subsumes most known lower
bound techniques for randomized communication complexity. How close
is it to being tight?    An upper bound on randomized communication
is proven by Massar~\cite{M01}.  We give a similar bound for entanglement assisted
communication complexity in terms of $\eff^*$.  Our bounds are stated for
zero-error communication complexity where the players may abort with some
probability $1-\eta$.  The weaker statement with $\epsilon$ error can
be derived using Lemma~\ref{lem:eta-epsilon}.

\begin{thm}\label{thm:upper-two-way}
For any distribution $\bp$ with outputs in $\A,\B$,
\begin{myenumerate} 
\item~\cite{M01} $R^{\eta,\parallel}(\bp) \leq \log(\frac{1}{1-\eta})\eff(\bp)\log(\#(\A\times\B))$
\item $R^{*,\eta,\parallel}(\bp) \leq \log(\frac{1}{1-\eta}){\eff^*(\bp)}\log(\#(\A\times \B))$
\item $R^{*,\eta}(\bp) \leq O\left(\sqrt{ \log(\frac{1}{1-\eta})\eff^*(\bp)}\right)$
\end{myenumerate}
\end{thm}
\begin{proof}
For the first item, let $P$ be a zero-error, zero-communication protocol with shared randomness for $\bp$
which has efficiency $\zeta=\frac{1}{\eff(\bp)}$.
Alice and Bob run the protocol  $N=\lceil\log(\frac{1}{1-\eta})\frac{1}{\zeta}\rceil$ times
and send their outcome to the referee in each run.
If the referee finds a valid run (where neither player aborts), he produces the corresponding outputs;
otherwise he aborts. 
Since each run has a probability $\zeta$ of producing a valid run,
the probability that the referee aborts is $(1-\zeta)^N\leq e^{-\zeta N}\leq 1-\eta$.

For the second item, the proof is the same but the players
share entanglement to run the protocol with shared entanglement
and efficiency $\frac{1}{\eff^*(\bp)}$.

If multiple rounds of communication are allowed, then a quadratic speedup is
possible in the quantum case by using a protocol for disjointness~\cite{BCW98,HW01,AA03}
on the input $u,v$ of length $N$, where $u_i$ is~0 if Alice aborts in the $i$th run and 1 
otherwise, similarly for $v$ with Bob.  
\end{proof}
\enlever{
If we could use the fact that a player's  probability of aborting is independent 
of that of the other player, that is, each player independently aborts a given run
with probability $\sqrt{\eta}$, then the above bounds can be improved
by a quadratic factor.  We can define a stronger bound $\eff^{\ind}$ which
adds this constraint while still remaining a lower bound on communication
complexity.  However this is a non-linear constraint.  Details will be
provided in the full version of this paper.
}

For one-way communication complexity,
the quantum partition bound is tight, up to arbitrarily small inefficiency.
We give the results for quantum communication since the rectangle bound
is already known to be tight for randomized communication complexity~\cite{JKN}.

\begin{thm}\label{thm:upper-one-way}
For any distribution $\bp$ and efficiency $\eta<1$, $Q^{*,\eta,\rightarrow}_{0}(\bp)\leq\frac{1}{2} \log(\eff^{*,\rightarrow}(\bp))+\log\log (1/(1-\eta))$.
\end{thm}
\begin{proof}
 Let $(\zeta,\bq)$ be an optimal solution for $\eff^{*,\rightarrow}(\bp)$. 
For any $x,y$, if we sample $a,b$ according to $\bq$, 
$\Pr_{\bq}[a\neq \bot|x]=\zeta$ and $\Pr_{\bq}[a,b|x,y]=\zeta p(a,b|x,y)$ 
for all $a,b\neq \bot$ and all $x,y$. 
Let Alice and Bob simulate this quantum distribution $N=\lceil\log(\frac{1}{1-\eta})\frac{1}{\zeta}\rceil$ times, keeping a record of the outputs $(a_i,b_i)$ for $i\in[N]$. Since this distribution is quantum, this requires no communication (only shared entanglement). Alice then communicates an index $i\in[N]$ such that $a_i\neq\bot$, if such an index exists, or just a random index if $a_i=\bot$ for all $i\in[N]$. Alice and Bob output $(a_i,b_i)$ corresponding to this index.

The correctness of the protocol follows from the fact that $\Pr_\bq[a_i=\bot(\forall i)]=(1-\zeta)^N\leq e^{-\zeta N}\leq 1-\eta$. The protocol then requires $\log N=-\log\zeta+\log\log(\frac{1}{1-\eta})$ bits of classical communication. Using superdense coding, this can be replaced by $\frac{1}{2}\log N$ qubits of quantum communication.
\end{proof}

\enlever{[Il faut sans doute déplacer le lemme suivant, mais je trouve que ça vaudrait la peine de le mettre quelque part.]} 
Finally, we show that $R^{\eta,\rightarrow}_{0}$ depends on $\eta$ by at most an additive constant.
The same is also true in the quantum model.
\begin{lemma}
For any distribution $\bp$ and efficiencies $0<\eta\leq\eta'<1$, $R^{\eta,\rightarrow}_{0}(\bp)\leq R^{\eta',\rightarrow}_{0}(\bp)\leq R^{\eta,\rightarrow}_{0}(\bp)-\log\eta+\log\log (1/(1-\eta'))$.
\end{lemma}
\enlever{
Note that in this one-way scenario, the reduced efficiency $\eta$ only affects Alice's output, while Bob's efficiency is always assumed to be perfect.}
\begin{proof}[Proof (sketch)]
 The proof is as above, except that we start from a  protocol for $\bp$ with efficiency $\eta$ instead of a quantum distribution. Note that Alice only needs to send the communication corresponding to the original protocol for the successful attempt.
\end{proof}

\section{Conclusion and open problems}

There are many questions to explore.  In experimental setups, 
in particular with optics, one
is faced with the very real problem that in most runs of
an experiment, no outcome is recorded.  The frequency
with which apparatus don't yield an outcome is called
detector inefficiency.
Can we find 
explicit Bell inequalities for quantum distributions 
that are very resistant to  detector inefficiency?
For experimental purposes, it is also important for the
distribution to be feasible to implement. 
One way to achieve this
could be to prove stronger bounds for the inequalities
based on the GHZ paradox given by Buhrman \textit{et al.}~\cite{BHMR06}.  
Their analysis is based on a tradeoff derived from the rectangle bound.
It may be possible to give sharper bounds with our
techniques.
Another is to consider 
asymmetric Bell inequalities and dimension witnesses~\cite{BPAGMS08,VPB10}.
Here, Alice prepares a state and Bob makes a measurement.
The goal is to have a Bell inequality demonstrating that
Alice's system has to be large. 
The dimension is exponential in the size of Alice's
message to Bob, so proving a lower bound on one-way communication complexity
gives a lower bound on the dimension.  
In order to close the detection loophole, one can also consider more realistic models of inefficiency, where the failure to produce a measurement outcome is the result of either the entangled state not being produced, or the detector of each player failing independently. 
This could be exploited by defining a stronger version of the partition/efficiency bound that also takes into account the probabilities of events where only one of the players produces a valid outcome.
While such a variation of the efficiency bound is meaningful for Bell tests, we have not considered it here as it might not be a lower bound on communication complexity.


\enlever{
We would like to see more applications.  For the Khot Vishnoi
distribution, we are not aware of any nontrivial upper bound so there
is a gap there to be improved.  
}

A family of lower bound techniques still not subsumed by the
efficiency bound are the information theoretic bounds such as
information complexity~\cite{CSWY01}.  It was recently shown that
information complexity is an upper bound on discrepancy~\cite{BW12},
and this upper bound was subsequently extended to a relaxation of the partition bound~\cite{KLL+12}. This \emph{relaxed} partition bound also subsumes most algebraic and combinatorial lower bound techniques, with the notable exception of the partition bound itself, and we would therefore like to see connections one way or the other between information complexity 
and the partition bound.

Finally, the quantum partition bound is of particular interest.  
It is hard to apply since it is not linear, and it amounts to 
finding a Tsirelson inequality, a harder
task to be sure than finding a good Bell inequality, that can nevertheless be approached via semidefinite programming relaxations~\cite{NPA08,DLTW08}.  On the
other hand, it is a very strong bound and one can hope to
get a better upper bound on quantum communication complexity.
Finding tight bounds 
complexity would be an important step to proving
the existence, or not, of exponential gaps for total boolean functions.

\section{Acknowledgements}
We wish to particularly thank Ronald de Wolf, Raghav Kulkarni and Iordanis Kerenidis for many
fruitful discussions.  Research funded in part by the EU grants QCS, QAlgo, ANR Jeune Chercheur CRYQ,
ANR Blanc QRAC and EU ANR Chist-ERA DIQIP.
J.R. acknowledges support from the action Mandats de Retour of the Politique Scientifique Fédérale Belge, and the Belgian ARC project COPHYMA.

\bibliography{prt-eff}

\appendix

\section {Proof of  Theorem~\ref{thm:partition}}
\label{apx:pf}

We give the proof that  for any distribution $\bp$,
$R_\epsilon^\eta(\bp) \geq \log(\prt_\epsilon^\eta(\bp))$.

\begin{proof}[Proof of Theorem~\ref{thm:partition}]
Let $\mathcal{P}$ be a protocol that simulates $\bp$ with $\eta$ detector efficiency
and $c$ bits of communication in the worst case, up to $\epsilon$ error in total variation distance.  Let $\bp'$ be the distribution produced by $\mathcal{P}$.
We can think of $\mathcal{P}$ as a probability distribution over
fully deterministic protocols $\{P_i\}$, where $P_i$ is chosen with probability $q(i)$.
Each deterministic protocol $P_i$ further decomposes into $2^c$ rectangles $\{R_{i,j}\}$
and in each rectangle, the players apply a local strategy $\bl_{i,j}$ defined over
inputs in $R$ and outputs in $\A\cup \{\bot\}\times \B\cup \{\bot\}$.

From this we  construct a feasible solution to the linear program for $\prt^\eta(\bp')$.
For any rectangle $R\subseteq X\times Y$ and any local distribution $l$
defined over inputs $R$ and outputs in $\A\cup \{\bot\}\times \B\cup \{\bot\}$,
we set 
	$$w_{R,\bl} = \sum_{i,j:R=R_{i,j} \text{ and } \bl=\bl_{i,j}} q(i).$$
Intuitively, $w_{R,\bl}$ is the probability of finding rectangle $R$ paired together with 
local strategy $\bl$ when choosing a deterministic protocol from $\mathcal{P}$.  
Each pair $R,l$ might appear in several of the deterministic protocols in 
$\mathcal{P}$, so we take the sum of the probabilities where this pair occurs.

First we claim that the objective function is $2^c$.  
 \begin{eqnarray*}
\sum_{R,\bl} w_{R,\bl}  &=&\sum_{R,\bl}  \sum_{i,j:R=R_{i,j} \text{ and } \bl=\bl_{i,j}} q(i) \\
	&=& \sum_{i,j} \sum_{R_{i,j},\bl_{i,j}} q(i)\\
	&=& 2^c \sum_i q(i)\\
	&=& 2^c.
\end{eqnarray*}

Now, we claim that all the constraints are verified.
For the first constraint, fix any $a,b,x,y\in \A\times \B\times \X\times \Y$.
By assumption, $\mathcal{P}$ outputs according to $p'(a,b|x,y)$, conditioned on
having output a value in $A\times B$.
Let us explicitly
calculate the (unconditional) probability that  $\mathcal{P}$ outputs $a,b$ on input $x,y$.
With probability $w_{R,\bl}$, $\mathcal{P}$ outputs according to the local
strategy $l$ applied on a rectangle $R$ containing $x,y$.  So the probability of outputting
$a,b$ is $\sum_{R: x,y\in R,\bl} w_{R,\bl} \cdot l(a,b|x,y)$.  The conditional probability is
obtained by dividing by the probability of outputting some $a',b' \in \A\times \B$ 
on input $x,y$.  This is precisely the quantity $\eta_{x,y}$.

The second constraint follows from the efficiency of $\mathcal{P}$.
This completes the proof.

\end{proof}

\section{Efficiency bound for protocols with bounded  efficiency}

\label{app:eff-eta}
In order to prove lower bounds on simulating $\bp$ with efficiency $\eta<1$, we define the following generalization of $\eff(\bp)$.
\begin{defn}
For any distribution $\bp$ with inputs in $X\times Y$ and outputs $A\times B$,
define $\eff^\eta(\bp)= 1/\zeta_\opt$, where $\zeta_\opt$ is the optimal value of the following
linear program.  The variables are $\zeta,\zeta_{xy}$ and $q_\bl$, where 
$\bl$ ranges over all local deterministic protocols with inputs taken from $\X\times\Y$
and outputs in $\A\cup\{\bot\}\times \B\cup\{\bot\}$.
\begin{align*}
   \zeta_\opt = & \max_{\zeta,\zeta_{xy},q_\bl\geq 0} && \zeta \\ 
&\text{subject to }&& \sum_{\bl\in\L_{\det}^\bot} q_\bl l(a,b|x,y)=\zeta_{xy} p(a,b|x,y)
		& \forall x,y,a,b\in \X{\times}\Y{\times}\A {\times}\B\\
& &&\sum_{\bl\in\L_{\det}^\bot} q_\bl=1
\\& &&\eta\zeta\leq\zeta_{xy}\leq\zeta& \forall x,y\in \X{\times}\Y.
\end{align*}
For randomized communication with error,
we define $\eff^\eta_\epsilon(\bp)=\min_{|p'-p|_1\leq \epsilon} \eff^\eta(\bp')$.
\end{defn}

This provides a lower bound for $R^\eta_\epsilon(\bp)$, which is equivalent to the lower bound obtained from $\prt^\eta_\epsilon(\bp)$ (we omit the proofs of these statements as they closely follow the lines of the special case $\eta=1$).
\begin{lemma}\label{lem:lower-eta}
 For any distribution $\bp$, we have $R^\eta_\epsilon(\bp)\geq\log\eff^\eta_\epsilon(\bp)$.
\end{lemma}
\begin{thm}\label{thm:equiv-eff-prt-eta}
For any distribution $\bp$,
$\eff^\eta_\epsilon(\bp)=\prt^\eta_\epsilon(\bp)$.
\end{thm}

We can also study the maximum $\eta$ such that $\bp$ can be simulated with efficiency $\eta_{xy}\geq \eta$ on input $x,y$, without any communication. We denote the inverse of this quantity by $\effl(\bp)$.
\begin{defn}
\label{def:nc-eff}
For any distribution $\bp$ with inputs in $X\times Y$ and outputs $A\times B$,
define $\effl(\bp)= 1/\eta$, where $\eta$ is the maximum $\eta$ such that $R^\eta(\bp)=0$.
\end{defn}

This quantity can be seen as a relaxation of $\eff(\bp)$, where we no longer
require the inefficiency to be the same for all inputs. Indeed, it can be rewritten as follows.
\begin{lemma}
\label{lem:nc-eff}
 For any distribution $\bp$, we have $\effl(\bp)=1/\zeta_\opt$, where $\zeta_\opt$ is the optimal value of the following
linear program.
\begin{align*}
   \zeta_\opt = & \max_{\zeta,\zeta_{xy},q_\bl\geq 0} && \zeta \\ 
&\text{subject to }&& \sum_{\bl\in\L_{\det}^\bot} q_\bl l(a,b|x,y)=\zeta_{xy} p(a,b|x,y)
		& \forall x,y,a,b\in \X{\times}\Y{\times}\A {\times}\B\\
& && \sum_{\bl\in\L_{\det}^\bot} q_\bl=1
\\&&& \zeta\leq\zeta_{xy}& \forall x,y\in \X{\times}\Y.
\end{align*}
\end{lemma}

By comparing the linear programs for the different quantities, we immediately obtain the following relations:
\begin{lemma}\label{lem:tradeoff}
 For any distribution $\bp$, we have $\eta\cdot \effl(\bp)\leq\eff^\eta(\bp)\leq\eta\cdot\eff(\bp)$.
\end{lemma}

\section{Lower bound for a Hidden Matching distribution}
\label{app:HM}
We first recall an application of KKL inequality as explained in \cite{DeWolf} which
we use in the proof.
\begin{lemma}
Let $A$ be a subset of $\{0,1\}^n$. Let $S$ be a subset of $\{1 \ldots n\}$. We define $\beta_S = \mathbb{E}_{x \in A} \left( (-1)^{S \cdot x} \right)$ where $S \cdot x = \sum_{i \in S} x_i$. Let $\mathcal{S}_2$ be the set of subsets of  $\{1 \ldots n\}$ of size 2. There exists an absolute constant $\mathcal{C}$ such that
$$\sum_{S \in \mathcal{S}_2} \beta_S^2 \leq \mathcal{C} \log \left( \frac{2^n}{|A|} \right)^2.$$
\end{lemma}

We now prove the theorem.
\begin{proof}[Proof of Theorem~\ref{thm:HM}]
Let $\bp'$ be such that  $|\bp' - \HM|_1 \leq \epsilon$. We lower bound $\eff^{\rightarrow}(\bp')$ using the dual of $\eff^{\rightarrow}$ (Lemma~\ref{lemma:eff-dual}).

\begin{align*}
\eff^{\rightarrow}(\bp')  = & \max_{B_{x,M,a,d,i,j} }    && \sum_{x,M,a,d,i,j} B_{x,M,a,d,i,j} \cdot p'(a,d,i,j|x,M) \\
 & \text{subject to } && \sum_{x\in X(\bl),M,a,d,i,j} B_{x,M,a,d,i,j} \cdot l(a,d,i,j|x,M) 
				\leq 1  &  \forall \bl\in\L_{\det}^\botA ,\\
\end{align*}
where we let $X(\ell)$ be the set of inputs for which Alice does not abort
when following the local deterministic strategy $\bl$.

We exhibit coefficients that satisfy the constraints and give us a good lower bound for the objective function for each $\bp'$ close to $\HM$.

To give an upper bound on the Bell value of any local deterministic strategy that may output $\bot$,
we will use the fact that such a strategy leads to a partition of Alice's inputs, where she doesn't abort, into rectangles. We will show an upper bound on the bias of each rectangle using the analysis from \cite{BRSW11}. However, in their analysis, the Bell value of the local strategy depends on the size of the rectangle, which will result in a poor upper bound. We will need to consider two different cases. If the rectangle is small enough, then we obtain a sufficiently good upper bound as is. If the rectangle is too big, we will need to subtract from the coefficients some constant that we will call $\mu$. Notice that in $\eff$, the constraint is that the Bell value of any local deterministic strategy is less than 1, but the absolute value is not bounded as in $\nu$. This is  why we can subtract without violating the constraint. The overall weight of those $\mu$ will not significantly affect the Bell value of a distribution close to the Hidden Matching distribution so the objective value will remain large.

Consider the following coefficients to the Bell functional.
\begin{align*}
B_{x,M,a,d,i,j} =&  \Phi'_{x,M,a,d,i,j} + \mu_{x,M} ,
\end{align*}
where
\begin{align*}
\mu_{x,M}   =&  - \frac{2^{\frac{\sqrt{n-1}}{2 \mathcal{C}}}}{n 2^{n+1} |\mathcal{M}_n|} \\
\Phi'_{x,M,a,d,i,j}  =&   \frac{2^{\frac{\sqrt{n - 1}}{2 \mathcal{C}}}}{n 2^n |\mathcal{M}_n|} \delta_{(i,j) \in M} \cdot (-1)^{\langle a,i \oplus j\rangle \oplus d \oplus x_i \oplus x_j}\\
\end{align*}
where $\delta$ is the Kronecker function,
and $\mathcal{M}_n$ is the set of matchings over edges $\{1, \ldots n \}$.

\paragraph{Verifying the constraints.}
Let $\bl \in \L_{det}^{\bot}$ and $X=X(\bl)$, the set of inputs for which Alice does not abort when following the local deterministic strategy $\bl$. The strategy $\bl$ partitions the set $X$ into $\bigcup_{a} X_a$ where Alice outputs $a$, and $\mathcal{M}_n$ into $ \bigcup_{d,i,j} R_{d,i,j}$ where Bob outputs $(d,i,j)$ because $\bl$ is local and deterministic.

First, we want to bound from above the value: 
\begin{eqnarray*}
 \sum_{x \in X, M, a, d, i ,j}   B_{x,M,a,d,i,j} \cdot l(a,d,i,j| x,M)
=  \sum_a \sum_{x \in X_a} \sum_{i,j,d} \sum_{M \in R_{i,j,d}}  B_{x,M,a,d,i,j}.
\end{eqnarray*}

We bound each term of the sum, for fixed $a$.

Let us first see what happens on small rectangles that is, when $X_a$ is small.
\begin{claim} 
If $|X_a| \leq 2^{n- \frac{\sqrt{n-1}}{2 \mathcal{C}}}$ then 
$ \sum_{x \in X_a} \sum_{i,j,d} \sum_{M \in R_{i,j,d}} B_{x,M,a,d,i,j} \leq \frac{1}{n}.$
\end{claim}

\begin{proof}
Since the $\mu_{x,M}$ are negative,

\begin{eqnarray*}
\sum_{x \in X_a} \sum_{i,j,d} \sum_{M \in R_{i,j,d}} B_{x,M,a,d,i,j} & = & \sum_{x \in X_a}  \sum_{i,j,d} \sum_{M \in R_{i,j,d}} \mu_{x,M} + \sum_{x \in X_a} \sum_{ i,j,d} \; \sum_{M \in R_{i,j,d}}  \Phi'_{x,M,a,d,i,j}\\ 
& \leq & \sum_{x \in X_a} \sum_{i,j,d} \sum_{M \in R_{i,j,d}}   \Phi'_{x,M,a,d,i,j}\\
& = & \frac{2^{\frac{\sqrt{n-1}}{2 \mathcal{C}}}}{n 2^n} \sum_{x \in X_a, M \in \mathcal{M}_n} \left( \sum_{d,i,j | l(a,d,i,j | x, M) = 1} \frac{(-1)^{x_i \oplus x_j \oplus d \oplus \langle a,i \oplus j\rangle}}{|\mathcal{M}_n|} \delta_{(i,j) \in M} \right)\\
& \leq & \frac{2^{\frac{\sqrt{n-1}}{2 \mathcal{C}}}}{n 2^n} |X_a|,
\end{eqnarray*}
where we have used the fact that there is exactly one tuple  $(d,i,j)$ such that $l(a,d,i,j|x,M) = 1$, because $l$ is deterministic and Bob doesn't abort.\\
Since $|X_a| \leq 2^{n- \frac{\sqrt{n-1}}{2 \mathcal{C}}}$ then this sum is less than $\frac{1}{n}$. 
\end{proof}

Now let us consider the case of the large rectangles.

\begin{claim} 
If $|X_a| \geq 2^{n- \frac{\sqrt{n-1}}{2 \mathcal{C}}}$ then
$ \sum_{x \in X_a} \sum_{i,j,d} \sum_{M \in R_{i,j,d}}  B_{x,M,a,d,i,j} \leq 0.$
\end{claim}

\begin{proof}
\begin{eqnarray*}
\sum_{x \in X_a} \sum_{i,j,d} \sum_{M \in R_{i,j,d}}  B_{x,M,a,d,i,j} & = &  \sum_{x \in X_a,M} \mu_{x,M} + \sum_{x \in X_a} \sum_{ i,j,d} \sum_{M \in R_{i,j,d}}  \Phi'_{x,M,a,d,i,j}\\
& = &- \frac{2^{\frac{\sqrt{n-1}}{2 \mathcal{C}}}}{n 2^{n+1}} |X_a| + \sum_{x \in X_a} \sum_{i ,j, d} \sum_{M \in R_{i,j,d}}  \Phi'_{x,M,a,d,i,j}.
\end{eqnarray*}

Let  $\beta_{i,j}^a = \mathbb{E}_{x \in X_a} ((-1)^{x_i \oplus x_j})$ and
$q_{i,j}^a = \sum_d \sum_{M\in R_{d,i,j}} \frac{(-1)^{\langle a,i \oplus j\rangle \oplus d}}{|\mathcal{M}_n|} \delta_{(i,j) \in M}$. Then
\begin{eqnarray*}
\sum_{x \in X_a} \sum_{i ,j, d} \sum_{M \in R_{i,j,d}}  \Phi'_{x,M,a,d,i,j} & = & \frac{2^{\frac{\sqrt{n-1}}{2 \mathcal{C}}}}{n 2^n}  \sum_{i,j} \sum_{x \in X_a} \sum_d \sum_{M \in R_{d,i,j}} \frac{(-1)^{x_i \oplus x_j \oplus d \oplus \langle a,i \oplus j\rangle}}{|\mathcal{M}_n|} \delta_{(i,j) \in M}\\
& =  & \frac{2^{\frac{\sqrt{n-1}}{2 \mathcal{C}}}}{n 2^n}  \sum_{i,j} |X_a| \beta_{i,j}^a \left( \sum_d \sum_{M\in R_{d,i,j}} \frac{(-1)^{\langle a,i \oplus j\rangle \oplus d}}{|\mathcal{M}_n|} \delta_{(i,j) \in M} \right)\\
& \leq & \frac{2^{\frac{\sqrt{n-1}}{2 \mathcal{C}}}}{n 2^n} |X_a| \sqrt{\sum_{i,j} |\beta_{i,j}^a|^2} \sqrt{\sum_{i,j} |q_{i,j}^a|^2}.\\
\end{eqnarray*}
The last line follows from the 
Cauchy-Schwarz inequality.

 On one hand,
\begin{eqnarray*}
|q_{i,j}^a| & \leq & \sum_d \sum_{M \in R_{d,i,j}} \frac{\delta_{(i,j) \in M}}{|\mathcal{M}_n|} =  \Prob_{M \in \mathcal{M}_n}( \text{$l$ outputs} (i,j) \in M)  \leq  \frac{1}{n-1},
\end{eqnarray*}
and $\sum_{i,j}|q_{i,j}^a| \leq 1$, so $ \sqrt{\sum_{i,j} |q_{i,j}^a|^2} \leq \frac{1}{\sqrt{n-1}}$.
On the other hand, the KKL inequality gives us (with $A = X_a$):
\begin{eqnarray*}
\sqrt{ \sum_{i,j} |\beta_{i,j}^a|^2} \leq \mathcal{C} \times \log \left(\frac{2^n}{|X_a|}\right).
\end{eqnarray*}

Hence,
\begin{eqnarray*}
\sum_{x \in X_a} \sum_{i ,j, d} \sum_{M \in R_{i,j,d}}  \Phi'_{x,M,a,d,i,j} & \leq & \frac{2^\frac{\sqrt{n-1}}{2 \mathcal{C}}}{n 2^n} |X_a| \mathcal{C} \log \left( \frac{2^n}{|X_a|} \right) \times \frac{1}{\sqrt{n -1}}\\
& \leq & \frac{2^\frac{\sqrt{n-1}}{2 \mathcal{C}}}{n 2^{n+1}} |X_a|,\\
\end{eqnarray*}
because $|X_a| \geq 2^{n- \frac{\sqrt{n-1}}{2 \mathcal{C}}}$ implies that $\mathcal{C} \log \left( \frac{2^n}{|X_a|} \right)  \frac{1}{\sqrt{n-1}}\leq \frac{1}{2}.$
\end{proof}

From Claims 1 and 2, we obtain:

\begin{eqnarray*}
\sum_{x \in X, M,a,d,i,j} B_{x,M,a,d,i,j} \cdot l(a,d,i,j,|,x,M)  \leq  \sum_{a | |X_a| \leq 2^{n - \frac{ \sqrt{n-1} } {2 \mathcal{C}}}} \frac{1}{n} \leq  1
\end{eqnarray*}

\paragraph{Value of the objective function.}
Let $\bp'$ be a distribution such that $|\bp' - \HM|_1 \leq \epsilon$.  
For any $x,M,a,d,i,j$, we define
$ \epsilon_{x,M,a,d,i,j} = |p'(a,d,i,j|x,M) - \HM(a,d,i,j|x,M)|$
and for any $x,M$, we have
$$   \sum_{a,d,i,j} \epsilon_{x,M,a,d,i,j} \leq \epsilon.$$

We want to lower bound 
$$\sum_{x,M,a,d,i,j} B_{x,M,a,d,i,j}  \cdot p'(a,d,i,j,| x,M).$$ 
Recall that  we have set 
$B_{x,M,a,d,i,j} =  \Phi'_{x,M,a,d,i,j} + \mu_{x,M}$.
We will consider the two terms separately.
Since $\bp'$ is a distribution, 
$$\sum_{x,M,a,d,i,j}  \mu_{x,M} \cdot p'(a,d,i,j|x,M) = \sum_{x,M} \mu_{x,M}= - \frac{2^{\frac{\sqrt{n-1}}{2 \mathcal{C}}}}{n 2^{n+1}} 2^n =   - \frac{2^{\frac{\sqrt{n-1}}{2 \mathcal{C}}}}{2n}$$
We also have 
\begin{eqnarray*}
\lefteqn{\sum_{x,M,a,d,i,j} \Phi'_{x,M,a,d,i,j}  \cdot p'(a,d,i,j,| x,M)}\\
& \geq & \frac{2^{\frac{\sqrt{n-1}}{2 \mathcal{C}}}}{n 2^n |\mathcal{M}_n|} \sum_{x,M}  \; \left(\sum_{a,d,i,j :  x_i \oplus x_j = d \oplus \langle a,i \oplus j\rangle} \delta_{(i,j) \in M} (\HM(a,d,i,j|x,M) - \epsilon_{x,M,a,d,i,j}) \right.\\
 && +    \left. \sum_{a,d,i,j :  x_i \oplus x_j \neq d \oplus \langle a,i \oplus j\rangle} \delta_{(i,j) \in M} ( - \HM(a,d,i,j|x,M) - \epsilon_{x,M,a,d,i,j} )\right)\\
& = & \frac{2^{\frac{\sqrt{n-1}}{2 \mathcal{C}}}}{n} - \frac{2^{\frac{\sqrt{n-1}}{2 \mathcal{C}}}}{n 2^n |\mathcal{M}_n|} \sum_{x,M,a,d,i,j} \epsilon_{x,M,a,d,i,j} \delta_{(i,j) \in M}\\
& \geq &  \frac{2^{\frac{\sqrt{n-1}}{2 \mathcal{C}}}}{n} - \frac{2^{\frac{\sqrt{n-1}}{2 \mathcal{C}}}}{n} \epsilon.
\end{eqnarray*}

Finally we get the value of the objective function 
$$\eff_\epsilon^\rightarrow(\HM) \geq {\sum_{x,M,a,d,i,j} B_{x,M,a,d,i,j} \cdot p'(a,d,i,j| x,M) }  \geq    \frac{2^{\frac{\sqrt{n-1}}{2 \mathcal{C}}}}{n} ( \frac{1}{2} - \epsilon).$$

\end{proof}

\end{document}